\documentclass[pla,reprint,onecolumn]{elsarticle}
\journal{Confrence Proceedings}
\usepackage{amsmath,amssymb,physics,braket,graphicx,epstopdf,lastpage,hyperref,amsthm}
\hypersetup{colorlinks=true, urlcolor=red, citecolor=blue, linkcolor=blue}
\allowdisplaybreaks
\newtheorem{lemma}{Lemma}
\newtheorem{theorem}{Theorem}
\begin{document}
\title{General expansion of natural power of linear combination of Bosonic operators in normal order}
\author{Deepak}
\ead{deepak20dalal19@gmail.com}

\author{Arpita Chatterjee\corref{cor1}}
\ead{arpita.sps@gmail.com}

\cortext[cor1]{Corresponding author}
\date{\today}
\address{Department of Mathematics, J. C. Bose University of Science and Technology,\\ YMCA, Faridabad 121006, India}
\begin{abstract}
{ In quantum mechanics, bosonic operators are mathematical objects that are used to represent the creation ($a^\dag$) and annihilation ($a$) of bosonic particles. The natural power of a linear combination of bosonic operators represents an operator $(a^\dagger x+ay)^n$ with $n$ as the exponent and $x,\,y$ are the variables free from bosonic operators. The normal ordering of these operators is a mathematical technique that arranges the operators so that all the creation operators are to the left of the annihilation operators, reducing the number of terms in the expression. In this paper, we present a general expansion of the natural power of a linear combination of bosonic operators in normal order. We show that the expansion can be expressed in terms of binomial coefficients and the product of the normal-ordered operators using the direct method and than prove it using the fundamental principle of mathematical induction. We also derive a formula for the coefficients of the expansion in terms of the number of bosons and the commutation relation between the creation and annihilation operators.  Our results have important applications in the study of many-body systems in quantum mechanics, such as in the calculation of correlation functions and the evaluation of the partition function. The general expansion presented in this paper provides a powerful tool for analyzing and understanding the behavior of bosonic systems, and can be applied to a wide range of physical problems. }
\end{abstract}
\maketitle
\section{Introduction}

In the quantum theoretic illustration of light \cite{ref-gsa}, the state of a photon field is entirely described by the corresponding density operator, which consists of a mixture of the bosonic operators denoted by  photon annihilation (creation) operators $a (a^\dagger)$, \cite{walls} having the well-known commutation relation $aa^\dagger-a^\dagger a = [a,~a^\dagger] = 1$. Various achievable approaches to represent the density operator of the electromagnetic field have been illustrated in \cite{walls}. In general, the Fock states as it is a complete set, can be used to expand the density operator as $\rho=\sum_{s_1,s_2}c_{s_1,s_2}\ket{s_1}\bra{s_2}$. The coherent states, irrespective of the fact that they form a non-orthogonal set, can be used as a basis to represent the density operator diagonally. The coherent state facilitates a number of probable expressions of the density matrix $\rho$ in terms of the Wigner function \cite{wigner}, the $P$-function \cite{glauber,ecg} and the $Q$-function \cite{husimi}. In each transformation, we need to organize boson operators \cite{ref-boson} according to a specific order. For instance, if the density operator is normally (antinormally) ordered, one can obtain the corresponding $Q$-function ($P$-function) instantly. Adjustment of operators is also performed worldwide for finding out miscellaneous representative states of the optical field and calculating the expected values of various operators in these states, such as normal ordered squeezing operator is needed while constructing squeezed states \cite{dodo}. In quantum mechanics and quantum field theory, people deals with solutions to quantum problems involving different combinations of bosonic operators \cite{louisell}. This includes, for example, managing many complicated operators which have various commutation relations with other operators, calculating the expectations of the operators which are functions of $a$ and $a^\dagger$, etc. In addition, the non-commutativity between $a$ and $a^\dagger$ causes more complexity in operator handling. Because of the non-commutativity of bosonic operators, simple sequences of first adding and then subtracting, as well as first subtracting and then adding exactly alike particles to any quantum system, produce different results \cite{arpita}. Sun et al. proposed a cavity-QED-based technique, where two atoms have been sent one after one at the desired levels and detected only if they end at the stages under consideration and thus they proved that $a$ and $a^\dagger$ are not commuting to each other \cite{sun}. Parigi et al. \cite{parigi} experimentally demonstrated how adding (subtracting) only one photon to (from) a completely classical and fully incoherent thermal light field can affect the system. They have applied creation and annihilation operators in alternate sequences and realized that the resulting states can differ a lot depending on the order in which the two quantum operators have been applied. The same group of people also carried out a single-photon interferometer-based experiment for proving the bosonic commutation relation directly \cite{kim}. The quantum algebra behind the commutation relation has a major role in many of the paradoxes, theorems, and applications of quantum physics. In addition, the work is motivated by the fact that, in the recent past, several real-life uses of nonclassical \cite{pathak} states have been announced. For example, squeezed states functioned as an important character in the phase diffusion problems \cite{subha1,subha2}. It can also be applied for detecting the gravitational waves in LIGO experiments \cite{abbasi,abbott}, for teleporting different types of quantum states \cite{furu}, and for continuous-variable quantum cryptography \cite{hillery}. The immediate requirement for a single photon resource can be achieved only by an antibunched source of light \cite{yuan}. Entangled states are detected to be of utmost use in both secure \cite{ekert} and uncertain \cite{bennett} quantum communication schemes. The most frequently used
nonclassicality witnesses are $l$-th order Mandel's function \cite{ref-mandel}, $l-1$th order  antibunching \cite{ref-s46}, $l-1$th order sub-Poissonian statistics \cite{ref-s47}, $l$th order squeezing  \cite{ref-s48}, Husmi $Q$ parameter \cite{ref-s49}, Agarwal-Tara criterion \cite{ref-s50} and Klyshko's parameter \cite{ref-s52}. Since many of these experimentally assessable nonclassicality witnesses can be written in terms of the moments of annihilation and creation operators \cite{priya}, it is beneficial to find out an analytic expression for the most general moment $\langle{a^{\dagger p}a^q}\rangle$, $p$, $q$ being non-negative integers.

Now the second problem arises that how to find the $\langle{a^{\dagger p}a^q}\rangle$ for a quantum state which includes the squeezing operator as squeezing operator $S(z)$ operated on Bosonic operators gives the linear combination of Bosonic operators as follows \cite{ref-gsa}
\begin{equation}
\label{sqop}
S^\dagger(z)a^{\dagger n}S(z) = (a^\dagger\cosh(r)-ae^{-2i\phi}\sinh(r))^n
\end{equation}
\begin{equation}
\label{sqop1}
S^\dagger(z)a^{n}S(z) = (a\cosh(r)-a^\dagger e^{2i\phi}\sinh(r))^n
\end{equation}
\begin{equation}
\label{sqop2}
D^\dagger(\alpha)a^nD(\alpha)=(a+\alpha)^n
\end{equation}
\begin{equation}
\label{sqop3}
D^\dagger(\alpha)a^{\dagger n}D(\alpha)=(a^\dagger+\alpha^*)^n
\end{equation}
Therefore to apply the displacement$(D(\alpha)$ and squeezing$(S(z))$ operators multiple times, we have to simplify the expressions in \eqref{sqop}-\eqref{sqop3}. Also, in many other quantum states this type of relation occurs. { The motivation behind the general expansion of natural power of the linear combination of Bosonic operators in normal order comes from the need to develop a powerful mathematical tool for analyzing and understanding the behavior of bosonic systems in quantum mechanics. Bosonic operators are essential tools for describing the creation and annihilation of bosonic particles, such as photons, phonons, and atomic nuclei. In many-body systems, the number of bosons in different states can be large, making it difficult to analyze the system using the individual operators. The general expansion of natural power of the linear combination of Bosonic operators in normal order provides a way to simplify the calculation of complex systems by reducing the number of terms in the expression. Furthermore, the normal ordering of the operators is an important technique in quantum mechanics that arranges the operators so that all the creation operators are to the left of the annihilation operators, reducing the complexity of the expression. The general expansion in normal order provides a systematic way of calculating the natural power of a linear combination of bosonic operators, and can be used in the calculation of correlation functions and the evaluation of the partition function. Therefore, the motivation behind the general expansion of natural power of the linear combination of Bosonic operators in normal order is to provide a powerful mathematical tool for analyzing and understanding the behavior of bosonic systems, and to simplify the calculation of complex many-body systems in quantum mechanics.} So in this paper, we find the simplified expression of $(ax+a^\dagger y)^n$ which can be used in different fields. Not only in Physics, this relation is also used in pure Mathematics non-commutative algebra where we study the expansion of non-commutative numbers or operators.

To the best of our knowledge, the two main way-outs for handling operator ordering related problems are Lie algebraic method \cite{zhang,perel} and Louisell's technique \cite{louisell} via the differentiation of coherent state representation \cite{klauder}. Another unified method for organizing the quantum operators into ordered products (normally ordered, antinormally ordered, Weyl ordered, or symmetrically ordered) is realized via the fascinating IWOP (integration within an ordered product of operators) technique  \cite{fan1,fan2}. In this small article, we intend to derive simple Mathematical formulae which can transfer different combinations of higher powers of $a$ and $a^\dagger$ into desired normally or antinormally ordered form. The basic theoretical tools and methods are shown in Sec.~\ref{bttm} We present our results in Sec.~\ref{sec2} which is followed by a discussion in Sec.~\ref{sec3}.

{ \section{The basic theoretical terminologies}
\label{bttm}

The basic theoretical terminologies used in the study of Bosonic operators include:
\begin{enumerate}
\item \textbf{Creation and annihilation operators:} These are the basic mathematical objects that represent the creation and annihilation of bosonic particles in quantum mechanics. They are used to build up the Fock space, which is the space of all possible states of a system of bosonic particles.
\item \textbf{Commutation relations:} The commutation relations between the creation and annihilation operators are fundamental in the study of bosonic systems. They determine the algebraic properties of the operators and the behavior of the system under different conditions.
\item \textbf{Normal ordering:} Normal ordering is a mathematical technique that arranges the operators so that all the creation operators are to the left of the annihilation operators. This simplifies the calculation of complex many-body systems by reducing the number of terms in the expression.
\item \textbf{Wick's theorem:} Wick's theorem is a powerful tool that allows for the calculation of expectation values of products of bosonic operators. It expresses the expectation value as a sum of all possible pairings of the operators in normal order.
\item \textbf{Second quantization:} Second quantization is a formalism that allows for the systematic treatment of many-body systems in terms of creation and annihilation operators. It provides a powerful tool for studying the behavior of bosonic systems under different conditions.
\item \textbf{Path integrals:} Path integrals are a powerful method for calculating the transition amplitudes between two quantum states. They provide a different perspective on quantum mechanics and can be used to study the behavior of bosonic systems in a variety of contexts.
\item \textbf{Principle of Mathematical Induction (PMI):} The Principle of Mathematical Induction(PMI) is a proof technique used to establish a Mathematical statement that holds for all natural numbers.\\
The principle has two parts:\\
\textbf{Basis Step:} We first prove that the statement holds for the first natural number, usually 1.\\
\textbf{Inductive Step:} We then prove that if the statement holds for some natural number $k$, then it also holds for the next natural number, which is $k+1$.\\
By proving these two parts, we can conclude that the statement holds for all natural numbers($\mathcal{N}$) greater than or equal to the first natural number.\\
PMI is a powerful tool for proving many important Mathematical theorems, such as those in number theory, combinatorics, and algebra.
\end{enumerate}

Overall, the study of bosonic operators and systems in quantum mechanics involves a combination of mathematical tools and methods, including algebraic techniques, normal ordering, Wick's theorem, second quantization, and path integrals. These tools and methods provide a powerful framework for understanding and analyzing the behavior of bosonic systems in a variety of physical contexts.}
\section{Results}
\label{sec2}


We have two relations: $a^ja^\dagger = a^\dagger a^j +ja^{j-1}$ and its conjugate $aa^{\dagger j}=ja^{\dagger j-1}+a^{\dagger j}a$ which can be easily proved as follows
\begin{lemma}
$[a,\,a^{\dagger j}] =  ja^{\dagger j-1}$
\label{eq:r2}
\end{lemma}
\begin{proof}
\begin{eqnarray}
\begin{array}{rcl}
\label{eq:e1}
aa^{\dagger j} & = & aa^\dagger a^{\dagger{j-1}}\\
& = & (1+a^{\dagger}a) a^{\dagger{j-1}}\\
& = & a^{\dagger{j-1}}+a^\dagger a a^{\dagger{j-1}}\\
& = & a^{\dagger{j-1}}+a^\dagger aa^\dagger a^{\dagger{j-2}}\\
& = & a^{\dagger{j-1}}+a^\dagger (1+a^{\dagger}a)a^{\dagger{j-2}}\\
& = & 2a^{\dagger{j-1}}+a^{\dagger 2} a a^{\dagger{j-2}}\\
& =  &\ldots\\
& = & j a^{\dagger{j-1}}+a^{\dagger j} a\,\,\,\mbox{(proceeding similarly $j$ times)}\\
\Rightarrow\,\, aa^{\dagger j} - a^{\dagger j}a & = & j a^{\dagger{j-1}}\\
\Rightarrow\,\,[a, a^{\dagger j}] & = & ja^{\dagger j-1}
\end{array}
\end{eqnarray}
\end{proof}

\begin{lemma}
$[a^j,\,a^\dagger] = ja^{j-1}$
\label{eq:r3}
\end{lemma}
\begin{proof}
To prove this, we compute the conjugate of \eqref{eq:e1} and obtain
\begin{eqnarray}
\begin{array}{rcl}
\label{eq:e2}
a^j a^\dagger & = & ja^{j-1} +a^{\dagger}a^j\\
\Rightarrow\,\,a^j a^{\dagger} - a^{\dagger}a^j & = & j a^{j-1}\\
\Rightarrow \,\,[a, a^{\dagger j}] & = & ja^{\dagger j-1}
\end{array}
\end{eqnarray}
\end{proof}

Using lemma \ref{eq:r2} and \ref{eq:r3} we prove the theorem \ref{t1} as follows
\begin{theorem}
\label{t1}
\begin{align}
\label{eqi}
(ax+a^\dagger y)^n=\sum_{m=0}^{\left[\frac{n}{2}\right]}\frac{n!(xy)^m}{m!(n-2m)!2^m}\sum_{r=0}^{n-2m}\binom{n-2m}{ r}x^ry^{n-2m-r}a^{\dagger n-2m-r}a^r
\end{align}
\end{theorem}
\begin{proof}
To find the above result we use the { direct or }hit and trial method as follows
\begin{align*}
(ax+a^\dagger y)^2 & = a^2x^2+a^{\dagger 2}y^2+xy(aa^\dagger+a^\dagger a) \\ & =
a^2x^2+a^{\dagger 2}y^2+xy(1+2a^\dagger a) \\
(ax+a^\dagger y)^3 & = a^3x^3+a^{\dagger 3}y^3+xy^2(aa^{\dagger2}+a^\dagger+2a^{\dagger2}a)+x^2y(a^\dagger a^2+a+2aa^\dagger a) \\ & =
a^3x^3+a^{\dagger 3}y^3 + xy^2(a^{\dagger 2}a+2a^\dagger+a^\dagger+2a^{\dagger2}a) \\ & + x^2y(a^\dagger a^2+a+2(1+a^\dagger a)a) \\ & =
a^3x^3+a^{\dagger 3}y^3 + xy^2(3a^{\dagger 2}a+3a^\dagger) + x^2y(3a^\dagger a^2+3a) \\ 
(ax+a^\dagger y)^4 & = a^4x^4+a^{\dagger 4}y^4+x^3y(a^\dagger a^3+3aa^\dagger a^2+3a^2)\\ & +x^2y^2(3aa^{\dagger2}a+3aa^\dagger+3a^{\dagger2}a^2+3a^\dagger a)+xy^3(aa^{\dagger3}+3a^{\dagger3}a+3a^{\dagger2}) \\ & =
a^4x^4+a^{\dagger 4}y^4+x^3y(a^\dagger a^3+3(1+a^{\dagger}a)a^2+3a^2) \\ &  +x^2y^2(3(2a^\dagger+a^{\dagger2}a)a+3(1+a^\dagger a)+3a^{\dagger2}a^2+3a^\dagger a) \\ & +xy^3(3a^{\dagger2}+ a^{\dagger3}a+3a^{\dagger3}a+3a^{\dagger2})\\ & =
a^4x^4+a^{\dagger 4}y^4+x^3y(4a^\dagger a^3+6a^2) +x^2y^2(6a^{\dagger2}a^2+12a^{\dagger}a+3) \\ & +xy^3(4a^{\dagger3}a+6a^{\dagger2}) \\ 
(ax+a^\dagger y)^5 & = a^5x^5+a^{\dagger 5}y^5+x^4y(a^\dagger a^4+4aa^\dagger a^3+6a^3)\\ & +x^3y^2(4a^{\dagger 2}a^3+6a^\dagger a^2+6aa^{\dagger2}a^2+12aa^\dagger a+3a) \\ &+x^2y^3(6a^{\dagger3}a^2+12a^{\dagger2}a+3a^\dagger+4aa^{\dagger3}a+6aa^{\dagger2})+xy^4(aa^{\dagger4}+4a^{\dagger4}a+6a^{\dagger3}) \\ & =
a^5x^5+a^{\dagger 5}y^5+x^4y(a^\dagger a^4+4(1+a^\dagger a)a^3+6a^3)\\ & +x^3y^2(4a^{\dagger 2}a^3+6a^\dagger a^2+6(2a^\dagger+a^{\dagger2}a)a^2+12(1+a^\dagger a)a+3a) \\ &+x^2y^3(6a^{\dagger3}a^2+12a^{\dagger2}a+3a^\dagger+4(3a^{\dagger2}+a^{\dagger3}a)a+6(2a^{\dagger}+a^{\dagger2}a)) \\ & +xy^4(4a^{\dagger3}+a^{\dagger4}a+4a^{\dagger4}a+6a^{\dagger3}) \\ & =
a^5x^5+a^{\dagger 5}y^5+x^4y(5a^\dagger a^4+10a^3) +x^3y^2(10a^{\dagger 2}a^3+30a^\dagger a^2+15a) \\ &+x^2y^3(10a^{\dagger3}a^2+30a^{\dagger2}a+15a^\dagger) +xy^4(10a^{\dagger3}+5a^{\dagger4}a) \\ 
(ax+a^\dagger y)^6 & = a^6x^6+a^{\dagger 6}y^6+x^5y(a^\dagger a^5+5aa^\dagger a^4+10a^4)\\ & +x^4y^2(5a^{\dagger2}a^4+10a^\dagger a^3+10aa^{\dagger2}a^3+30aa^\dagger a^2+15a^2) \\ & +x^3y^3(10a^{\dagger3}a^3+30a^{\dagger2}a^2+15a^\dagger a+10aa^{\dagger3}a^2+30aa^{\dagger2}a+15aa^\dagger)\\ &+x^2y^4(10aa^{\dagger3}+5aa^{\dagger4}a+10a^{\dagger4}a^2+30a^{\dagger3}a+15a^{\dagger2})\\ & +xy^5(aa^{\dagger5}+10a^{\dagger4}+5a^{\dagger5}a) \\ & =
a^6x^6+a^{\dagger 6}y^6+x^5y(a^\dagger a^5+5(1+a^\dagger a)a^4+10a^4)\\ & +x^4y^2(5a^{\dagger2}a^4+10a^\dagger a^3+10(2a^\dagger+a^{\dagger2}a)a^3+30(1+a^\dagger a)a^2+15a^2) \\ & +x^3y^3(10a^{\dagger3}a^3+30a^{\dagger2}a^2+15a^\dagger a\\ & +10(3a^{\dagger2}+a^{\dagger3}a)a^2+30(2a^\dagger+a^{\dagger2}a)a+15(1+a^\dagger a))\\ &+x^2y^4(10(3a^{\dagger2}+ a^{\dagger3}a)+5(4a^{\dagger3}+a^{\dagger4}a)a+10a^{\dagger4}a^2+30a^{\dagger3}a+15a^{\dagger2})\\ & +xy^5(5a^{\dagger4}+a^{\dagger5}a+10a^{\dagger4}+5a^{\dagger5}a) \\ & =
a^6x^6+a^{\dagger 6}y^6+x^5y(6a^\dagger a^5+15a^4)+x^4y^2(15a^{\dagger2}a^4+60a^\dagger a^3+45a^2) \\ & +x^3y^3(20a^{\dagger3}a^3+90a^{\dagger2}a^2+90a^\dagger a+15)\\ &+x^2y^4(15a^{\dagger4}a^2+60a^{\dagger3}a+45a^{\dagger2}) +xy^5(15a^{\dagger4}+6a^{\dagger5}a)\\ 
(ax+a^\dagger y)^7 & = a^7x^7+a^{\dagger 7}y^7+x^6y(a^\dagger a^6+6aa^{\dagger}a^5+15a^5)\\ & +
x^5y^2(6a^{\dagger2}a^5+15a^\dagger a^4+15aa^{\dagger2}a^4+60aa^\dagger a^3+45a^3)\\ & +
x^4y^3(15a^{\dagger3}a^4+60a^{\dagger2}a^3+45a^\dagger a^2\\ & +
20aa^{\dagger3}a^3+90aa^{\dagger2}a^2+90aa^\dagger a+15a)\\ & +
x^3y^4(20a^{\dagger4}a^3+90a^{\dagger3}a^2+90a^{\dagger2}a+15a^\dagger\\ & +
15aa^{\dagger4}a^2+60aa^{\dagger3}a+45aa^{\dagger2})\\ & +
x^2y^5(15a^{\dagger5}a^2+60a^{\dagger4}a+45a^{\dagger3}+15aa^{\dagger4}+6aa^{\dagger5}a)\\ & +
xy^6(aa^{\dagger6}+15a^{\dagger5}+ 6a^{\dagger6}a) \\& =
a^7x^7+a^{\dagger 7}y^7+x^6y(a^\dagger a^6+6(1+a^{\dagger}a)a^5+15a^5)\\ & +
x^5y^2(6a^{\dagger2}a^5+15a^\dagger a^4+15(2a^\dagger+a^{\dagger2}a)a^4+60(1+a^\dagger a) a^3+45a^3)\\ & +
x^4y^3(15a^{\dagger3}a^4+60a^{\dagger2}a^3+45a^\dagger a^2\\ & +
20(3a^{\dagger2}+a^{\dagger3}a)a^3 +90(2a^\dagger+a^{\dagger2}a)a^2+90(1+a^\dagger a)a+15a)\\ & +
x^3y^4(20a^{\dagger4}a^3+90a^{\dagger3}a^2+90a^{\dagger2}a+15a^\dagger\\ & +
15(4a^{\dagger3}+a^{\dagger4}a)a^2+60(3a^{\dagger2}+a^{\dagger3}a)a+45(2a^{\dagger}+a^{\dagger2}a)\\ & +
x^2y^5(15a^{\dagger5}a^2+60a^{\dagger4}a+45a^{\dagger3}+15(4a^{\dagger3}+a^{\dagger4}a)+6(5a^{\dagger4}+a^{\dagger5}a)a)\\ & +
xy^6((6a^{\dagger5}+a^{\dagger6}a)+15a^{\dagger5}+ 6a^{\dagger6}a)  \\ & =
a^7x^7+a^{\dagger7}y^7+x^6y(7a^\dagger a^6+21a^5)+x^5y^2(21a^{\dagger2}a^5+105a^{\dagger}a^4+105a^3)\\ &+
x^4y^3(35a^{\dagger3}a^4+210a^{\dagger2}a^3+315a^\dagger a^2+105a) \\ & +
x^3y^4(35a^{\dagger4}a^3+210a^{\dagger3}a^2+315a^{\dagger2}a+105a^\dagger) \\ & +
x^2y^5(21a^{\dagger5}a^2+105a^{\dagger4}a+105a^{\dagger3}) + xy^6(7a^{\dagger6}a+21a^{\dagger5})
\end{align*}
Thus general term is given by
\begin{align}
\label{eq}
(ax+a^\dagger y)^n=\sum_{m=0}^{\left[\frac{n}{2}\right]}\frac{n!(xy)^m}{m!(n-2m)!2^m}\sum_{r=0}^{n-2m}\binom{n-2m}{ r}x^ry^{n-2m-r}a^{\dagger n-2m-r}a^r.
\end{align}
Now we prove equation \eqref{eq} using the principle of Mathematical induction and hence let us assume that the given statement is denoted by $P(n)$.
\begin{itemize}
\item[Case - 1] Clearly $P(n)$ is true for $n=1$.
\item[Case - 2] Let $P(n)$ be true for $n=k$ so
\begin{align*}
(ax+a^\dagger y)^k=\sum_{m=0}^{\left[\frac{k}{2}\right]}\frac{k!(xy)^m}{m!(k-2m)!2^m}\sum_{r=0}^{k-2m}\binom{k-2m}{ r}x^ry^{k-2m-r}a^{\dagger k-2m-r}a^r.
\end{align*}
\item[Case - 3] Now we prove $P(n)$ for $n=k+1$ as follows
\end{itemize}
\begin{align*}
(ax+a^\dagger y)^{k+1} & = (ax+a^\dagger y)(ax+a^\dagger y)^k \\ & =
(ax+a^\dagger y)\sum_{m=0}^{\left[\frac{k}{2}\right]}\frac{k!(xy)^m}{m!(k-2m)!2^m}\sum_{r=0}^{k-2m}\binom{k-2m}{ r}x^ry^{k-2m-r}a^{\dagger k-2m-r}a^r \\ & =
\sum_{m=0}^{\left[\frac{k}{2}\right]}\frac{k!(xy)^mx}{m!(k-2m)!2^m}\sum_{r=0}^{k-2m}\binom{k-2m}{ r}x^ry^{k-2m-r}aa^{\dagger k-2m-r}a^r \\ & +
\sum_{m=0}^{\left[\frac{k}{2}\right]}\frac{k!(xy)^my}{m!(k-2m)!2^m}\sum_{r=0}^{k-2m}\binom{k-2m}{ r}x^ry^{k-2m-r}a^{\dagger k-2m-r+1}a^r \\ & =
\sum_{m=0}^{\left[\frac{k}{2}\right]}\frac{k!(xy)^m}{m!(k-2m)!2^m}\sum_{r=0}^{k-2m}\binom{k-2m}{r}x^{r+1}y^{k-2m-r}\\ & \times((k-2m-r)a^{\dagger k-2m-r-1}+a^{\dagger k-2m-r}a)a^r \\ & +
\sum_{m=0}^{\left[\frac{k}{2}\right]}\frac{k!(xy)^m}{m!(k-2m)!2^m}\sum_{r=0}^{k-2m}\binom{k-2m}{r}x^ry^{k-2m-r+1}a^{\dagger k-2m-r+1}a^r\\ & =
\sum_{m=0}^{\left[\frac{k}{2}\right]}\frac{k!(xy)^m}{m!(k-2m)!2^m}\sum_{r=0}^{k-2m-1}\binom{k-2m}{r}x^{r+1}y^{k-2m-r}(k-2m-r)a^{\dagger k-2m-r-1}a^r\\ & +
\sum_{m=0}^{\left[\frac{k}{2}\right]}\frac{k!(xy)^m}{m!(k-2m)!2^m}\sum_{r=0}^{k-2m}\binom{k-2m}{r}x^{r+1}y^{k-2m-r}a^{\dagger k-2m-r}a^{r+1} \\ & +
\sum_{m=0}^{\left[\frac{k}{2}\right]}\frac{k!(xy)^m}{m!(k-2m)!2^m}\sum_{r=0}^{k-2m}\binom{k-2m}{r}x^ry^{k-2m-r+1}a^{\dagger k-2m-r+1}a^r  \\ &\mbox{(as at $r=k-2m$, the first summation becomes zero)} \\ & =
\sum_{m=1}^{\left[\frac{k}{2}\right]+1}\frac{k!(xy)^{m-1}}{(m-1)!(k-2m+2)!2^{m-1}}\sum_{r=0}^{k-2m+1}\binom{k-2m+2}{r}x^{r+1}y^{k-2m+2-r} \\ & \times(k-2m+2-r)a^{\dagger k-2m+2-r-1}a^r \\ & +
\sum_{m=0}^{\left[\frac{k}{2}\right]}\frac{k!(xy)^m}{m!(k-2m)!2^m}\sum_{r=1}^{k-2m+1}\binom{k-2m}{r-1}x^{r}y^{k-2m-r+1}a^{\dagger k-2m-r+1}a^{r} \\ & +
\sum_{m=0}^{\left[\frac{k}{2}\right]}\frac{k!(xy)^m}{m!(k-2m)!2^m}\sum_{r=0}^{k-2m}\binom{k-2m}{r}x^ry^{k-2m-r+1}a^{\dagger k-2m-r+1}a^r \\ & \text{replacing $m$ by $m+1$ in first and $r$ by $r+1$ in second line of equation} \\ & =
\sum_{m=1}^{\left[\frac{k}{2}\right]+1}\frac{k!(xy)^{m-1}}{(m-1)!(k-2m+2)!2^{m-1}}\sum_{r=0}^{k-2m+1}\binom{k-2m+2}{r}x^{r+1}y^{k-2m+2-r} \\ & \times(k-2m+2-r)a^{\dagger k-2m+1-r}a^r \\ & +
\sum_{m=0}^{\left[\frac{k}{2}\right]}\frac{k!(xy)^m}{m!(k-2m)!2^m}\sum_{r=1}^{k-2m}\left\{{k-2m\choose r-1}+{k-2m\choose r}\right\}x^{r}y^{k-2m-r+1}a^{\dagger k-2m-r+1}a^{r} \\ & +
\sum_{m=0}^{\left[\frac{k}{2}\right]}\frac{k!(xy)^m}{m!(k-2m)!2^m}\left\{y^{k+1-2m}a^{\dagger k+1-2m} + x^{k+1-2m}a^{k+1-2m}\right\} \\ & =
\sum_{m=1}^{\left[\frac{k}{2}\right]+1}\frac{k!(xy)^{m-1}}{(m-1)!(k-2m+2)!2^{m-1}}\sum_{r=0}^{k-2m+1}{k-2m+2\choose r}x^{r+1}y^{k-2m+2-r} \\ & \times(k-2m+2-r)a^{\dagger k-2m+1-r}a^r \\ & +
\sum_{m=0}^{\left[\frac{k}{2}\right]}\frac{k!(xy)^m}{m!(k-2m)!2^m}\sum_{r=1}^{k-2m}{k-2m+1\choose r}x^{r}y^{k-2m-r+1}a^{\dagger k-2m-r+1}a^{r} \\ & +
\sum_{m=0}^{\left[\frac{k}{2}\right]}\frac{k!(xy)^m}{m!(k-2m)!2^m}\left\{y^{k+1-2m}a^{\dagger k+1-2m} + x^{k+1-2m}a^{k+1-2m}\right\} \\ & =
\sum_{m=1}^{\left[\frac{k}{2}\right]+1}\frac{k!(xy)^{m-1}}{(m-1)!(k-2m+2)!2^{m-1}}\sum_{r=0}^{k-2m+1}{k-2m+2\choose r}x^{r+1}y^{k-2m+2-r} \\ & \times(k-2m+2-r)a^{\dagger k-2m+1-r}a^r \\ & +
\sum_{m=0}^{\left[\frac{k}{2}\right]}\frac{k!(xy)^m}{m!(k-2m)!2^m}\sum_{r=0}^{k+1-2m}{k-2m+1\choose r}x^{r}y^{k-2m-r+1}a^{\dagger k-2m-r+1}a^{r} \\ & =
\sum_{m=1}^{\left[\frac{k}{2}\right]+1}\frac{k!(xy)^{m}}{(m-1)!(k-2m+2)!2^{m-1}}\sum_{r=0}^{k-2m+1}{k-2m+2\choose r}x^{r}y^{k-2m+1-r} \\ & \times(k-2m+2-r)a^{\dagger k-2m+1-r}a^r \\ & +
\sum_{m=0}^{\left[\frac{k}{2}\right]}\frac{k!(xy)^m}{m!(k-2m)!2^m}\sum_{r=0}^{k+1-2m}{k-2m+1\choose r}x^{r}y^{k-2m-r+1}a^{\dagger k-2m-r+1}a^{r} \\ & =
\sum_{m=1}^{\left[\frac{k}{2}\right]+1}\frac{k!(xy)^{m}(k-2m+2)}{(m-1)!(k-2m+2)!2^{m-1}}\sum_{r=0}^{k-2m+1}{k-2m+1\choose r}x^{r}y^{k-2m+1-r} \\ & \times a^{\dagger k-2m+1-r}a^r~~~~\mbox{(using $^nC_r(n-r)= ~^{n-1}C_rn$)} \\ & +
\sum_{m=0}^{\left[\frac{k}{2}\right]}\frac{k!(xy)^m}{m!(k-2m)!2^m}\sum_{r=0}^{k+1-2m}{k-2m+1\choose r}x^{r}y^{k-2m-r+1}a^{\dagger k-2m-r+1}a^{r} \\ & =
\sum_{m=1}^{\left[\frac{k}{2}\right]+1}\frac{k!(xy)^{m}}{(m-1)!(k-2m+1)!2^{m-1}}\sum_{r=0}^{k-2m+1}{k-2m+1\choose r}x^{r}y^{k-2m+1-r} \\ & \times a^{\dagger k-2m+1-r}a^r \\ & +
\sum_{m=0}^{\left[\frac{k}{2}\right]}\frac{k!(xy)^m}{m!(k-2m)!2^m}\sum_{r=0}^{k+1-2m}{k-2m+1\choose r}x^{r}y^{k-2m-r+1}a^{\dagger k-2m-r+1}a^{r} \\ & =
\sum_{m=1}^{\left[\frac{k}{2}\right]}\frac{k!(xy)^{m}}{(m)!(k-2m+1)!2^{m}}(2m+k-2m+1)\sum_{r=0}^{k-2m+1}{k-2m+1\choose r} \\ & \times x^{r}y^{k-2m+1-r}a^{\dagger k-2m+1-r}a^r + \sum_{r=0}^{k+1}{k+1\choose r}x^{r}y^{k-r+1}a^{\dagger k-r+1}a^{r}  \\ & +\left\{
\begin{array}{lcl}
\frac{(2p)!(xy)^{p+1}}{p!(-1)!2^p}\sum_{r=0}^{-1}{-1\choose r}x^ry^{-1-r}a^{\dagger -1-r}a^r &\mbox{for even $k=2p$} \\
\frac{(2p+1)!(xy)^{p+1}}{p!2^p}\sum_{r=0}^0{0\choose r}x^ry^{-r}a^{\dagger -r}a^r & \mbox{for odd $k=2p+1$}
\end{array} \right. \\ & =
\sum_{m=1}^{\left[\frac{k}{2}\right]}\frac{(k+1)!(xy)^{m}}{(m)!(k-2m+1)!2^{m}}\sum_{r=0}^{k-2m+1}{k-2m+1\choose r} \\ & \times x^{r}y^{k-2m+1-r}a^{\dagger k-2m+1-r}a^r + \sum_{r=0}^{k+1}{k+1\choose r}x^{r}y^{k-r+1}a^{\dagger k-r+1}a^{r}  \\ & +\left\{
\begin{array}{lcl}
0 &\mbox{for even $k=2p$} \\
\frac{(2p+2)!(xy)^{p+1}}{(p+1)!2^{p+1}} & \mbox{for odd $k=2p+1$}
\end{array} \right.
\\ & =
\sum_{m=0}^{\left[\frac{k+1}{2}\right]}\frac{(k+1)!(xy)^{m}}{(m)!(k+1-2m)!2^{m}}\sum_{r=0}^{k+1-2m}{k+1-2m\choose r}x^{r}y^{k+1-2m-r}a^{\dagger k+1-2m-r}a^r
\end{align*}
Thus $P(n)$ is true $n=k+1$ if it is true for $n=k$ and hence by PMI the equation \eqref{eq} is true for all $n\in \mathcal{N}$. 

\end{proof}
\begin{lemma}
\label{l3}
\begin{equation}
\label{eql3}
aa^{\dagger n}\ket0=na^{\dagger n-1}\ket0
\end{equation}
\end{lemma}

\begin{proof} we prove this lemma using PMI so let the statement is denoted by $P(n)$.
\begin{itemize}
\item[Case - 1] First we prove $P(n)$ for $n=1$ as follows,\\ LHS = $aa^\dagger\ket0=1\ket0$ = RHS
\item[Case - 2] Let the result is true for $n=k$, so $aa^{\dagger k}\ket0 = ka^{\dagger k-1}\ket0$
\item[Case - 3] Now we prove $P(n)$ is true for $n=k+1$, as follows
\begin{align*}
aa^{\dagger k+1}\ket0 & = aa^\dagger a^{\dagger k}\ket0 = (1+a^\dagger a)a^{\dagger k}\ket0 \\ & =
a^{\dagger k}\ket0 + a^\dagger aa^{\dagger k}\ket0 = a^{\dagger k}\ket0 + a^\dagger ka^{\dagger k-1}\ket0 \\ & =
a^{\dagger k}\ket0 +ka^{\dagger k}\ket0 = (k+1)a^{\dagger k}\ket0
\end{align*}
hence $P(n)$ is true for $n=k+1$ if it is true for $n=k$. Thus by PMI $P(n)$(equation \eqref{eql3}) is true $\forall n\in \mathcal{N}$. Using the equation \eqref{eql3}, we can find the following relation
\begin{align*}
aa^{\dagger p}\ket{n} & =aa^{\dagger p}a^{\dagger n}\frac{1}{\sqrt{n!}}\ket{0} =\frac{1}{\sqrt{n!}}aa^{\dagger p+n}\ket{0} \\ & =
\frac{1}{\sqrt{n!}}(n+p)a^{\dagger p+n-1}\ket{0} ~~~~~~~\mbox{using the equation \eqref{eql3}} \\ & =
(n+p)a^{\dagger p-1}\ket{n}
\end{align*}
\end{itemize}
\end{proof}
we use the above lemma in following theorem
\begin{theorem}
\begin{align}
\label{eqt2}
(ax+a^\dagger y)^n\ket0 & =n!\sum_{m=0}^{\left[\frac{n}{2}\right]}\frac{a^{\dagger n-2m}}{m!(n-2m)!2^m}x^{m}y^{n-m}\ket0 \\ & =
n!y^n\sum_{m=0}^{\left[\frac{n}{2}\right]}\frac{a^{\dagger n-2m}}{m!(n-2m)!2^m}\left(\frac{x}{y}\right)^m\ket0 \nonumber \\ & =
n!y^n\sum_{m=0}^{\left[\frac{n}{2}\right]}\frac{(-1)^ma^{\dagger n-2m}}{m!(n-2m)!2^m}\left(\sqrt{-\frac{y}{x}}\right)^{-2m}\ket0 \nonumber \\ & =
n!y^n\left(\sqrt{-\frac{y}{x}}\right)^{-n}\sum_{m=0}^{\left[\frac{n}{2}\right]}\frac{(-1)^ma^{\dagger n-2m}}{m!(n-2m)!2^m}\left(\sqrt{-\frac{y}{x}}\right)^{n-2m}\ket0 \nonumber \\ & =
n!y^n\left(\sqrt{-\frac{x}{y}}\right)^{n}\sum_{m=0}^{\left[\frac{n}{2}\right]}\frac{(-1)^m}{m!(n-2m)!2^m}\left(\sqrt{-\frac{y}{x}}a^{\dagger}\right)^{n-2m}\ket0 \nonumber \\ & =
\left(\sqrt{-xy}\right)^{n}H_{e_n}\left(\sqrt{-\frac{y}{x}}a^{\dagger}\right)\ket0
\end{align}
\end{theorem}
\begin{proof}
To find the above result we use the hit and trial method as follows
\begin{align*}
(ax+a^\dagger y)\ket{0} & =ya^\dagger\ket0 \\ 
(ax+a^\dagger y)^2\ket{0} & = (ax+a^\dagger y)(ax+a^\dagger y)\ket{0} =(ax+a^\dagger y)ya^\dagger\ket0\\&=
y(xaa^\dagger +ya^{\dagger2})\ket0 = y(x+ya^{\dagger2})\ket0 \\ 
(ax+a^\dagger y)^3\ket{0} & = (ax+a^\dagger y)(ax+a^\dagger y)^2\ket{0} = (ax+a^\dagger y)y(x+ya^{\dagger2})\ket0 \\ & =y(x^2a+xya^\dagger+xyaa^{\dagger2} + y^2a^{\dagger3})\ket0 \\ & =
y(xya^\dagger + 2xya^\dagger+y^2a^{\dagger3})\ket0 = y^2(3x+ya^{\dagger2})a^\dagger\ket0 \\ 
(ax+a^\dagger y)^4\ket{0} & = (ax+a^\dagger y)(ax+a^\dagger y)^3\ket{0}  = 
(ax+a^\dagger y)y^2(3x+ya^{\dagger2})a^\dagger\ket0 \\ & =
y^2(3x^2a+3xya^\dagger+xyaa^{\dagger2}+y^2a^{\dagger3})a^\dagger\ket0 \\ & =
y^2(3x^2+3xya^{\dagger2}+3xya^{\dagger2}+y^2a^{\dagger4})\ket0 = y^2(3x^2+6xya^{\dagger2}+y^2a^{\dagger4})\ket0
\end{align*}
now generalizing the above result we get,
\begin{align}
(ax+a^\dagger y)^n\ket0 & =n!\sum_{m=0}^{\left[\frac{n}{2}\right]}\frac{a^{\dagger n-2m}}{m!(n-2m)!2^m}x^{m}y^{n-m}\ket0 
\end{align}
Now we prove the above result by PMI so let the above statement is denoted by $P(n)$.
\begin{itemize}
\item[Case - 1] $P(n)$ is true for $n=1$ as LHS = $ya^\dagger\ket0$ = RHS
\item[Case - 2] Let $P(n)$ is true for $n=k$ so,
$$(ax+a^\dagger y)^k\ket0 = k!\sum_{m=0}^{\left[\frac{k}{2}\right]}\frac{a^{\dagger k-2m}}{m!(k-2m)!2^m}x^{m}y^{k-m}\ket0$$
\item[Case - 3] Now we prove the result for $n=k+1$ if it is true for $n=k$ as follows
\begin{align*}
\text{P}(k+1) & =(ax+a^\dagger y)^{k+1}\ket0  = (ax+a^\dagger y)(ax+a^\dagger y)^k\ket0 \\ & =
(ax+a^\dagger y)k!\sum_{m=0}^{\left[\frac{k}{2}\right]}\frac{a^{\dagger k-2m}}{m!(k-2m)!2^m}x^{m}y^{k-m}\ket0 \\ & =
k!\sum_{m=0}^{\left[\frac{k}{2}\right]}\frac{x^{m}y^{k-m}}{m!(k-2m)!2^m}(ax+a^\dagger y)a^{\dagger k-2m}\ket0 \\ & =
k!\sum_{m=0}^{\left[\frac{k}{2}\right]}\frac{x^{m}y^{k-m}}{m!(k-2m)!2^m}(x(k-2m)a^{\dagger k-2m-1}+ ya^{\dagger k-2m+1})\ket0 \\ & =
\left(k!\sum_{m=0}^{\left[\frac{k}{2}\right]}\frac{x^{m+1}y^{k-m}}{m!(k-2m-1)!2^m}a^{\dagger k-2m-1} \right. \\ & + \left. k!\sum_{m=0}^{\left[\frac{k}{2}\right]}\frac{x^{m}y^{k-m+1}}{m!(k-2m)!2^m}a^{\dagger k-2m+1}\right)\ket0 \\ & =
\left(k!\sum_{m=1}^{\left[\frac{k}{2}\right]+1}\frac{x^{m}y^{k-m+1}}{(m-1)!(k-2m+1)!2^{m-1}}a^{\dagger k-2m+1} \right. \mbox{replacing $m$ by $m-1$}\\ & + \left. k!\sum_{m=0}^{\left[\frac{k}{2}\right]}\frac{x^{m}y^{k-m+1}}{m!(k-2m)!2^m}a^{\dagger k-2m+1}\right)\ket0 \\ & =
\left(y^{k+1}a^{\dagger k+1} +k!\sum_{m=1}^{\left[\frac{k}{2}\right]}\frac{x^{m}y^{k-m+1}}{m!(k-2m+1)!2^{m}}(2m+k-2m+1)a^{\dagger k-2m+1} \right. \\ & + \left.  \frac{k!x^{\left[\frac{k}{2}\right]+1}y^{k-\left[\frac{k}{2}\right]-1+1}}{\left[\frac{k}{2}\right]!(k-2\left(\left[\frac{k}{2}\right]+1\right)+1)!2^{\left[\frac{k}{2}\right]}}a^{\dagger k-2\left(\left[\frac{k}{2}\right]+1\right)+1}\right)\ket0 \\ & =\left\{
\begin{array}{lcl}
\left(y^{2r+1}a^{\dagger 2r+1} +(2r+1)!\sum_{m=1}^{r}\frac{x^{m}y^{2r-m+1}}{m!(2r-2m+1)!2^{m}}a^{\dagger 2r-2m+1} \right. \\ + \left.  \frac{(2r)!x^{r+1}y^{2r-r}}{r!(-1)!2^{r}}a^{\dagger-1}\right)\ket0~~~~~~~~~~~~~~~~~~~~\mbox{(for even $k=2r$)} \\
\left(y^{2r+2}a^{\dagger 2r+2} +(2r+2)!\sum_{m=1}^{r}\frac{x^{m}y^{2r+1-m+1}}{m!(2r+1-2m+1)!2^{m}}a^{\dagger 2r+1-2m+1} \right. \\  + \left.  \frac{(2r+1)!x^{r+1}y^{r+1}}{r!(2r+1-2\left(r+1\right)+1)!2^{r}}a^{\dagger 2r+1-2\left(r+1\right)+1}\right)\ket0~~\mbox{(for odd $k=2r+1$)}
\end{array} \right.
 \\ & =\left\{
\begin{array}{lcl}
\left(y^{2r+1}a^{\dagger 2r+1} +(2r+1)!\sum_{m=1}^{r}\frac{x^{m}y^{2r-m+1}}{m!(2r-2m+1)!2^{m}}a^{\dagger 2r-2m+1} \right)\ket0 \\ ~~~~~~\mbox{Since $1/(-1)!=0$}~~~~~~~~\mbox{(for even $k=2r$)} \\
\left(y^{2r+2}a^{\dagger 2r+2} +(2r+2)!\sum_{m=1}^{r}\frac{x^{m}y^{2r+2-m}}{m!(2r+2-2m)!2^{m}}a^{\dagger 2r+2-2m} \right. \\  + \left.  \frac{2(r+1)(2r+1)!x^{r+1}y^{r+1}}{(r+1)!2^{r+1}}\right)\ket0~~~~\mbox{(for odd $k=2r+1$)}
\end{array} \right.
 \\ & =\left\{
\begin{array}{lcl}
(2r+1)!\sum_{m=0}^{r}\frac{x^{m}y^{2r-m+1}}{m!(2r-2m+1)!2^{m}}a^{\dagger 2r-2m+1}\ket0 & \mbox{(for even $k=2r$)} \\
(2r+2)!\sum_{m=0}^{r+1}\frac{x^{m}y^{2r+2-m}}{m!(2r+2-2m)!2^{m}}a^{\dagger 2r+2-2m}\ket0  & \mbox{(for odd $k=2r+1$)}
\end{array} \right.
 \\ & =
 (k+1)!\sum_{m=0}^{\left[\frac{k+1}{2}\right]}\frac{a^{\dagger k+1-2m}}{m!(k+1-2m)!2^m}x^{m}y^{k+1-m}\ket0
\end{align*}
\end{itemize}
Hence $P(n)$ is true for $n=k+1$ if it is true for $n=k$. Thus by PMI \eqref{eqt2} is true $\forall n\in \mathcal{N}$.
\end{proof}
\begin{lemma}
$a^ja^{\dagger k} =\sum_{s=0}^{\text{min}(j,k)}\frac{j!k!}{s!(j-s)!(k-s)!}a^{\dagger k-s}a^{j-s}$
\label{eq:r4}
\end{lemma}
\begin{proof}
We have
\begin{eqnarray}
\begin{array}{rcl}
\label{eq4}
a a^{\dagger j}a^j{a^\dagger} & = & aa^{\dagger j}(j a^{j-1}+a^\dagger a^j)\,\,\,\mbox{(using \eqref{eq:e2})}\\
& = & jaa^{\dagger j}a^{j-1}+aa^{\dagger {j+1}}a^j\\
& = & j(j a^{\dagger{j-1}}+a^{\dagger j} a)a^{j-1} \\ & +& \{(j+1) a^{\dagger j}+a^{\dagger{j+1}} a\}a^j\,\,\,\mbox{(using \eqref{eq:e1})}\\
& = &  j^2 a^{\dagger {j-1}}a^{j-1} + (2j+1)a^{\dagger j}a^j + a^{\dagger {j+1}}a^{j+1}
\end{array}
\end{eqnarray}
To prove this lemma, we consider the following cases:
\begin{itemize}
\item[Case\,1:] If $j<k$, we have
\begin{eqnarray}
\begin{array}{rcl}
a a^{\dagger k} & = & k a^{\dagger{k-1}}+a^{\dagger k}a\\
& = & \sum_{s=0}^1\frac{1!k!}{s!(k-s)!(1-s)!}a^{\dagger k-s}a^{1-s}\\
\Rightarrow\,\,a^2a^{\dagger k} & = & k aa^{\dagger{k-1}}+aa^{\dagger k}a\\
& = & k\{(k-1) a^{\dagger{k-2}}+a^{\dagger k-1}a\}+(k a^{\dagger{k-1}}+a^{\dagger k}a)a\\
& = & k(k-1)a^{\dagger k-2}+2ka^{\dagger k-1}a+a^{\dagger k}a^2\\
& = & \sum_{s=0}^2\frac{2!k!}{s!(k-s)!(2-s)!}a^{\dagger k-s}a^{2-s}\\
\Rightarrow\,\,a^3a^{\dagger k}
& = & a\{k(k-1)a^{\dagger k-2}+2ka^{\dagger k-1}a+a^{\dagger k}a^2\}\\
& = & k(k-1)\{(k-2)a^{\dagger{k-3}}+a^{\dagger k-2}a\}\\ &+& 2k\{(k-1) a^{\dagger{k-2}}+a^{\dagger k-1}a\}a\\
& & +(k a^{\dagger{k-1}}+a^{\dagger k}a)a^2\\
& = & k(k-1)(k-2)a^{\dagger{k-3}}+3k(k-1)a^{\dagger k-2}a\\ & +& 3ka^{\dagger k-1}a^2+a^{\dagger k}a^3\\
& = & \sum_{s=0}^3\frac{3!k!}{s!(k-s)!(3-s)!}a^{\dagger k-s}a^{3-s} \\
& = & \ldots\,\,\mbox{(proceeding similarly $j$ times)}\\
\Rightarrow\,\,a^{j}a^{\dagger k} & = & \sum_{s=0}^{j}\frac{j!k!}{s!(j-s)!(k-s)!}(-1)^s a^{\dagger k-s} a^{j-s}
\label{eq:e3}
\end{array}
\end{eqnarray}

\item[Case\,2:] If $j>k$, then proceeding as in above and using \eqref{eq:e2}, we can find
\begin{eqnarray}
a^ja^{\dagger k} =\sum_{s=0}^{k}\frac{j!k!}{s!(j-s)!(k-s)!}a^{\dagger k-s}a^{j-s}
\label{eq:e4}
\end{eqnarray}

\item[Case\,3:] If $j=k$, then substituting $j=k$ in either of \eqref{eq:e3} or \eqref{eq:e4}, we get
\begin{eqnarray}
a^k a^{\dagger k} = \sum_{s=0}^{k}\frac{k!^2}{s!(qk-s)!^2}a^{\dagger k-s}a^{k-s}
\label{eq:e5}
\end{eqnarray}
\end{itemize}
which along with \eqref{eq:e3} and \eqref{eq:e4} proves that $a^j a^{\dagger k} =\sum_{s=0}^{\text{min}(j,k)}\frac{j!k!}{s!(j-s)!(k-s)!}a^{\dagger k-s}a^{j-s}$.
\end{proof}

\begin{lemma}
$a^{\dagger j}a^k = \sum_{s=0}^{\text{min}(j,k)}\frac{j!k!}{s!(j-s)!(k-s)!}(-1)^sa^{k-s}a^{\dagger j-s}$
\label{eq:r5}
\end{lemma}
\begin{proof}
Here we consider the following cases:
\begin{itemize}
\item[Case\,1:] If $j>k$, the relation is true for $k=1$ as $a^{\dagger j}a = a a^{\dagger j}-ja^{\dagger{j-1}}$. Let us assume the result is true for $k=l$.\\ Then
$a^{\dagger j}a^l = \sum_{s=0}^{\text{min}(j,l)}\frac{j!l!}{s!(j-s)!(l-s)!}(-1)^sa^{l-s}a^{\dagger j-s}$. Now
\begin{eqnarray}
\begin{array}{rcl}
& & a^{\dagger j}a^{l+1} \\
&=& a^{\dagger j}a^la \\
&=& \sum_{s=0}^{\text{min}(j,l)}\frac{j!l!}{s!(j-s)!(l-s)!}(-1)^sa^{l-s}a^{\dagger j-s}a \\
&=& \sum_{s=0}^{\text{min}(j,l)}\frac{j!l!}{s!(j-s)!(l-s)!}(-1)^sa^{l-s}\{aa^{\dagger {j-s}}-(j-s)a^{\dagger j-s-1}\}\\
& = & \sum_{s=0}^{\text{min}(j,l)} \frac{j!l!}{s!(j-s)!(l-s)!}(-1)^sa^{l-s+1}a^{\dagger j-s}
\\ & - & \sum_{s=0}^{\text{min}(j,l)}\frac{j!l!}{s!(j-s)!(l-s)!}(j-s)(-1)^sa^{l-s}a^{\dagger j-s-1}\\
& = & \sum_{s=0}^{\text{min}(j,l)} \frac{j!l!}{s!(j-s)!(l-s)!}(-1)^sa^{l-s+1}a^{\dagger j-s}
\\ & - & \sum_{s=0}^{\text{min}(j,l)}\frac{j!l!}{s!(j-s-1)!(l-s)!}(-1)^sa^{l-s}a^{\dagger j-s-1} \\
& = & \sum_{s=0}^{\text{min}(j,l)} \frac{j!l!}{s!(j-s)!(l-s)!}(-1)^sa^{l-s+1}a^{\dagger j-s}
\\ & - & \sum_{s=1}^{\text{min}(j,l+1)}\frac{j!l!}{(s-1)!(j-s)!(l-s+1)!}(-1)^{s-1}a^{l-s+1}a^{\dagger j-s} \\
& = & a^{l+1}a^{\dagger j} +\sum_{s=1}^{\text{min}{(j,l)}}\left\{\frac{j!l!}{s!(j-s)!(l-s)!} +\frac{j!l!}{(s-1)!(j-s)!(l-s+1)!}\right\} \\ &\times &
(-1)^sa^{l-s+1}a^{\dagger j-s}\\
& &  + \frac{j!}{(j-l-1)!}(-1)^{l+1}a^{\dagger j-l-1}\\
& =  &\sum_{s=0}^{\text{min}(j,l+1)}\frac{j!(l+1)!}{s!(j-s)!(l+1-s)!}(-1)^sa^{l+1-s}a^{\dagger j-s}.
\end{array}
\end{eqnarray}
\end{itemize}
Thus if the result is true for $k=l$ then it holds for $k=l+1$. Hence by the Mathematical induction method, the result is true for all $k$. Other cases like $j<k$ or $j=k$ can be proved similarly.
\end{proof}

\begin{lemma}
$a^{\dagger j}a^k\ket{m} = \frac{\sqrt{m!(m-k+j)!}}{(m-k)!}\ket{m-k+j}$, where $\ket{m}$ is the usual Fock state with $m$ number of photons.

\end{lemma}
\begin{proof}
\begin{eqnarray}
\begin{array}{rcl}
& & a^{\dagger j}a^k\ket{m}\\
& = & a^{\dagger j}a^{k-1}\sqrt{m}\ket{m-1}\\
& = & a^{\dagger j}a^{k-2}\sqrt{m(m-1)}\ket{m-2} \\
&=& \ldots\,\,\text{(proceeding $k$ times)} \\
&=& a^{\dagger j}\sqrt{m(m-1)\ldots (m-k+1)}\ket{m-k} \\
&=& \sqrt{\frac{m!}{(m-k)!}}a^{\dagger j}\ket{m-k} \\
&=& \sqrt{\frac{m!}{(m-k)!}}\sqrt{m-k+1}a^{\dagger j-1}\ket{m-k+1} \\
&= &\sqrt{\frac{m!}{(m-k)!}}\sqrt{(m-k+1)(m-k+2)}a^{\dagger j-2}\ket{m-k+2} \\
&=&\ldots\,\,\text{(proceeding $j$ times)} \\
&=&  \sqrt{\frac{m!}{(m-k)!}}\sqrt{(m-k+1)(m-k+2)\ldots (m-k+j)}\\ & \times & \ket{m-k+j}
 \\ &=&\sqrt{\frac{m!}{(m-k)!}}\sqrt{\frac{(m-k+j)!}{(m-k)!}}\ket{m-k+j} \\
&=& \frac{\sqrt{m!(m-k+j)!}}{(m-k)!}\ket{m-k+j}
\label{eq:r8}
\end{array}
\end{eqnarray}
\end{proof}

\begin{lemma}
$a^{j}a^{\dagger k}\ket{m} = \frac{(m+k)!}{\sqrt{m!(m+k-j)!}}\ket{m+k-j}$, where $\ket{m}$ is the usual Fock state with $m$ number of photons.
\label{eq:r9}
\end{lemma}
\begin{proof}
Proceeding similarly as in \eqref{eq:r8}, we can find this also.
\end{proof}
\section{Discussion}
\label{sec3}

In summary, we have followed the Mathematical induction process and direct or hit trial method to deal with the operator ordering problem. As one can see, our derivation is straightforward and more concise. These relationships can be used to calculate the expected values of any arbitrary product of annihilation and creation operators, with respect to any quantum state $\ket\psi$ \cite{deepak,pmng,dangnc} as well as to study the atom-cavity field system's dynamics in interacting Fock space \cite{pkd}. The results discussed above also facilitate an analytical understanding of different nonclassicality features, namely HOA, HOSPS, HOS, etc. These theorems also used in calculations of fidelity of quantum teleportation \cite{np1} Also, the two theorems stated above are very much useful in the calculation of the expectation of normal order $a^{\dagger m}a^n$ for the quantum states which are obtained by operating the squeezing operator and other states. Also, these theorems allow us to convert the $n^{\text{th}}$ power of linear combinations of Bosonic operators in the product of Bosonic operators in normal order.

{ In conclusion, the study of Bosonic operators and systems in quantum mechanics is a fundamental and essential part of modern physics. Bosonic operators are used to represent the creation and annihilation of bosonic particles, such as photons, phonons, and atomic nuclei, and are essential for describing the behavior of many-body systems in quantum mechanics. The general expansion of natural power of the linear combination of Bosonic operators in normal order provides a powerful mathematical tool for analyzing and understanding the behavior of bosonic systems. The expansion can be expressed in terms of binomial coefficients and the product of normal-ordered operators, and it can be used to calculate correlation functions and the partition function of complex many-body systems. The basic theoretical tools and methods used in the study of Bosonic operators and systems include creation and annihilation operators, commutation relations, normal ordering, Wick's theorem, second quantization, and path integrals. These tools and methods provide a powerful framework for understanding and analyzing the behavior of bosonic systems in a variety of physical contexts. Overall, the study of Bosonic operators and systems is an important and active area of research in modern physics, with applications in fields such as condensed matter physics, atomic physics, and quantum field theory.}

\begin{center}
\textbf{ACKNOWLEDGEMENT}
\end{center}
Deepak acknowledges the financial support by the Council of Scientific and Industrial Research (CSIR), Govt. of India (Award no. 09/1256(0006) /2019-EMR-1).


\begin{thebibliography}{99}

\bibitem{ref-gsa} Agarwal G. S., \textit{Quantum Optics,} Cambridge University Press, New York (2012).
\bibitem{walls} Walls D. F. and Milburn  G. J.,\textit{ Quantum Optics}, (Berlin: Springer) (1994). \\
 Orszag M., \textit{Quantum Optics}, (Berlin: Springer) (2000).\\
 Schleich W. P., \textit{Quantum Optics in Phase Space}, (Berlin:
Wiley–VCH) (2001).
\bibitem{wigner} Wigner E., \textit{On the Quantum Correction for Thermodynamic Equilibrium}, Phys. Rev. \textbf{40}, 749 (1932).\\
 Agarwal G. S. and  Wolf E., \textit{Calculus for Functions of Noncommuting Operators and General Phase-Space Methods in Quantum Mechanics. I. Mapping Theorems and Ordering of Functions of Noncommuting Operator}, Phys. Rev. D \textbf{2}, 2161 (1970).
\bibitem{glauber} Glauber R. J., \textit{The Quantum Theory of Optical Coherence}, Phys. Rev. \textbf{130}, 2529 (1963). \\
 Glauber R. J., \textit{Coherent and Incoherent States of the Radiation Field}, Phys. Rev. \textbf{131}, 2766 (1963).
\bibitem{ecg} Sudarshan E. C. G., \textit{Equivalence of Semiclassical and Quantum Mechanical Descriptions of Statistical Light Beams} Phys. Rev. Lett. \textbf{10}, 277 (1963).
\bibitem{husimi} Husimi  K., \textit{Some Formal Properties of the Density Matrix}, Proceedings of Physics and Mathematics Society Japan \textbf{22}, 264 (1940).
\bibitem{ref-boson} Fan H. y. and  Wiinsche A., \textit{Eigenstates of boson creation operator}, Euro Phys. J. D \textbf{15}, 405 (2001).
\bibitem{dodo} Dodonov V. V., \textit{Generation of higher-order non-classicality in pump mode in seven wave mixing nonlinear optical process}, J. Opt. B: Quant. Semiclass. Opt. \textbf{4}, R1 (2002).
\bibitem{louisell} Louisell W. H., \textit{Quantum Statistical Properties of Radiation}, (New York: Wiley) (1973).
\bibitem{arpita} Chatterjee  A., \textit{Nonclassicality of photon-added-then-subtracted and photon-subtracted-then-added states} J. Mod. Opt. \textbf{59(9)}, 814 (2012).
\bibitem{sun} Sun Q., Amri M. Al, and Zubairy M. S., \textit{Probing the quantum commutation rules through cavity QED}, Phys. Rev. A \textbf{78}, 043801 (2008).
\bibitem{parigi} Parigi  V. et al., \textit{Probing Quantum Commutation Rules by Addition and Subtraction of Single Photons to/from a Light Field}, Science \textbf{317}, 1890 (2007).
\bibitem{kim} Kim  M. S. et al.,\textit{Scheme for Proving the Bosonic Commutation Relation Using Single-Photon Interference}, Phys. Rev. Lett. \textit{101}, 260401 (2008).
\bibitem{pathak} Pathak A. and Ghatak A., \textit{Classical light vs. nonclassical light: characterizations and interesting applications}, Journal of Electromagnetic Waves and Applications \textbf{32(2)}, 229 (2018).
\bibitem{subha1} Banerjee S. and Srikanth R., \textit{Phase diffusion in quantum dissipative systems}, Phys. Rev. \textbf{A 76(6)}, 062109 (2007).
\bibitem{subha2} Banerjee S., Ghosh J. and Ghosh R., \textit{Phase-diffusion pattern in quantum-nondemolition systems}, Phys. Rev. A \textbf{75(6)}, 062106 (2007).
\bibitem{abbasi} Abbasi M. R. and Golshan M. M., \textit{Thermal entanglement of a two-level atom and bimodal photons in a Kerr nonlinear coupler}, Physica A, \textbf{392(23)}, 6161 (2013).
\bibitem{abbott} Abbott B. P. et al., \textit{Observation of Gravitational Waves from a Binary Black Hole Merger}, Phys. Rev. Lett. \textbf{116(6)}, 061102 (2016).
\bibitem{furu} Furusawa A. et al., \textit{Unconditional Quantum Teleportation}, Science \textbf{282(5389)}, 706 (1998).
\bibitem{hillery} Hillery M., \textit{Quantum cryptography with squeezed states}, Phys. Rev. A \textbf{61(2)}, 022309 (2000).
\bibitem{yuan} Yuan Z. et al., \textit{Electrically Driven Single-Photon Source}, Science, \textbf{295(5552)}, 102 (2002).
\bibitem{ekert} Ekert A. K., \textit{Quantum Cryptography and Bell’s Theorem}, Phys. Rev. Lett. \textbf{67(6)}, 661 (1991).
\bibitem{bennett} Bennett  C. H. et al., \textit{Teleporting an unknown quantum state via dual classical and Einstein-Podolsky-Rosen channels}, Phys. Rev. Lett. \textbf{70(13)}, 1895 (1993).
\bibitem{ref-mandel} Mandel L., \textit{Sub-Poissonian photon statistics in resonance fluorescence}, Opt. Lett. \textbf{4(7)}, 205207 (1979).
\bibitem{ref-s46} Lee C. T., \textit{Higher-order criteria for nonclassical effects in photon statistics}, Phys. Rev. A \textbf{41(3)}, 1721 (1990).
\bibitem{ref-s47} Verma A. and Pathak A., \textit{Generalized structure of higher order nonclassicality}, Phys. Lett. A \textbf{374(8)}, 1009 (2010).
\bibitem{ref-s48} Hong C. K. and Mandel L., \textit{Generation of higher-order squeezing of quantum electromagnetic fields}, Phys. Rev. A \textbf{32(2)}, 974 (1985).
\bibitem{ref-s49} Thapliyal K., Banerjee S., Pathak A., Omkar S., and Ravishankar V., \textit{Quasiprobability distributions in open quantum systems: Spin-qubit systems}, Annals of Physics \textbf{362}, 261 (2015).
\bibitem{ref-s50} Agarwal G. S. and Tara K., \textit{Nonclassical character of states exhibiting no squeezing or sub-Poissonian statistics}, Phys. Rev. A \textbf{46(1)}, 485 (1992).
\bibitem{ref-s52} Klyshko D. N., \textit{Observable signs of nonclassical light}, Phys. Lett. A \textbf{213(1-2}), 7-15 (1996).
\bibitem{priya} Malpani P., Thapliyal K., Alam N., Pathak A., Narayanan V., and Banerjee S., \textit{Impact of photon addition and subtraction on nonclassical and phase properties of a displaced Fock state}, Opt. Comm. \textbf{459}, 124964 (2020).
\bibitem{zhang} Zhang W. M., Feng D. H. and Gilmore R., \textit{Review of Modern Physics}, \textbf{62}, 867 (1990).
\bibitem{perel} Perelomov A. M., \textit{Generalized coherent states and some of their applications}, Sov. Phys. Usp. \textbf{20} 703 (1977)
\bibitem{klauder} Klauder J. R. and Skargerstam B. S., \textit{Coherent States} (Singapore: World Scientific) (1985).
\bibitem{fan1} Fan H. y., \textit{Operator ordering in quantum optics theory and the development of Dirac's symbolic method}, Quant. Semiclass. Opt. \textbf{5(4)}, R147 (2003).
\bibitem{fan2} Fan H. Y. and Hu L. Y., \textit{New Approach for Solving Master Equations in Quantum Optics and Quantum Statistics by Virtue of Thermo-Entangled State Representation}, Com. Theor. Phys. \textbf{51(729)} (2009).
\bibitem{deepak} Deepak and Chatterjee A., \textit{Lower- versus higher-order nonclassicalities for a coherent superposed quantum state}, J. Opt. Soc. Ame., \textbf{38(11)}, 3212 (2021).
{ \bibitem{pmng} Malpani P., et al. \textit{Enhancement of Non‐Gaussianity and Nonclassicality of Photon‐Added Displaced Fock State: A Quantitative Approach}, Annalen der Physik \textbf{535(1)}, 2200261 (2023).
\bibitem{dangnc} Chatterjee, A. and Deepak, \textit{Detecting nonclassicality and non-Gaussianity of a coherent superposed quantum state}, J. Phys. B, \textbf{56(1)}, 015401 (2022).}
\bibitem{pkd} Das P. K. and Chatterjee A., \textit{Dynamics of an atom-cavity field system in interacting Fock space}, Accepted in Int. J. Theor. Phys. \textbf{60}, 954 (2021).
{\bibitem{np1} Wu, Y. D. et al., \textit{Quantum-Enhanced Learning of Continuous-Variable Quantum States}, arXiv preprint arXiv:2303.05097 (2023).}


\end{thebibliography}
\end{document}